%% file: main.tex
\documentclass[a4paper, 10pt, english, twocolumn]{article}

\input{packages}
\input{macros}

\begin{document}

\input{content/title}
\input{content/abstract}
\input{content/intro}
\input{content/background}
\input{content/defs}
\input{content/frame/_frame}

\input{content/solution}
\input{content/anal}
\input{content/summ}
\input{content/references}
\input{content/appendix}

\end{document}

%% file: packages.tex
\usepackage[utf8]{inputenc}
\usepackage[doi=true]{biblatex}
\usepackage{hyperref}
\usepackage{amsmath}
\usepackage{amsthm} 
\usepackage{amsfonts} 
\usepackage{amssymb} 
\usepackage{mathtools} 
\usepackage{relsize}
\usepackage[shortlabels]{enumitem}
\usepackage{titlesec}
\usepackage[english]{babel}
\usepackage{graphicx}
\usepackage{float}

\usepackage{caption}
\usepackage{url}
\usepackage{setspace}

\usepackage[linesnumbered,algoruled,boxed,lined]{algorithm2e}

\usepackage{subfig} 

\usepackage{csquotes}

\usepackage{times}
\usepackage{verbatim}

\usepackage{titlesec}
\usepackage{enumitem}

\usepackage{footnote}
\usepackage{footmisc}

\usepackage[toc,page]{appendix}

\DefineBibliographyStrings{english}{%
urlseen = {Retrieved},
}

%% file: macros.tex
\theoremstyle{definition}
\newtheorem{definition}{Definition}[section]

\newtheorem{theorem}{Theorem}

\addto\captionsenglish{
  \renewcommand{\contentsname}%
    {Table of Contents}%
}

\newcommand{\bs}[1][]{
    \ifthenelse{\equal{#1}{}}{>_{S_O}}{>_{S_#1}}
}

\newcommand{\multifunction}[4] {
        #1 = 
\begin{cases}
    #2,          & #3\\
    #4,          & \text{otherwise}
\end{cases}
}

\AtBeginDocument{}
\AtBeginDocument{}
\AtBeginDocument{}

\newcounter{savefootnote}
\newcounter{symfootnote}
\newcommand{\symfootnote}[1]{%
   \setcounter{savefootnote}{\value{footnote}}%
   \setcounter{footnote}{\value{symfootnote}}%
   \ifnum\value{footnote}>8\setcounter{footnote}{0}\fi%
   \let\oldthefootnote=\thefootnote%
   \renewcommand{\thefootnote}{\fnsymbol{footnote}}%
   \footnote{#1}%
   \let\thefootnote=\oldthefootnote%
   \setcounter{symfootnote}{\value{footnote}}%
   \setcounter{footnote}{\value{savefootnote}}%
}

\titlespacing*{\section}
{0cm}
{0.2cm}
{0cm}
\titlespacing*{\subsection}
{0cm}
{0.1cm}
{0cm}
\titlespacing*{\subsubsection}
{0cm}
{0cm}
{0cm}
\setlist[itemize]{nosep}
\setlist[enumerate]{nosep}

%% file: content/title.tex
\title{Ranking and benchmarking framework for sampling algorithms on synthetic data streams\symfootnote{This work was supported by the construction EFOP-3.6.3-VEKOP-16-2017-00002. The project was supported by the European Union, co-financed by the European Social Fund.}}

\author{
    József Dániel Gáspár 
    \and Martin Horváth
    \and Győző Horváth
    \and Zoltán Zvara
}

\date{June 17, 2020}

\maketitle

%% file: content/abstract.tex
\begin{abstract}
    In the fields of \textbf{big data}, \textbf{AI}, and \textbf{streaming processing}, we work with large amounts of data from multiple sources.
    Due to memory and network limitations, we process data streams on distributed systems to alleviate computational and network loads.
    When data streams with non-uniform distributions are processed, we often observe overloaded partitions due to the use of simple hash partitioning.
    To tackle this imbalance, we can use dynamic partitioning algorithms that require a sampling algorithm to precisely estimate the underlying distribution of the data stream.
    There is no standardized way to test these algorithms.
    We offer an extensible ranking framework with benchmark and hyperparameter optimization capabilities and supply our framework with a data generator that can handle concept drifts.
    Our work includes a generator for dynamic micro-bursts that we can apply to any data stream.
    We provide algorithms that react to concept drifts and compare those against the state-of-the-art algorithms using our framework.

\end{abstract}

%% file: content/intro.tex
\section{Introduction}
The number of computers grows at a never before seen rate, mainly due to the inclusion of microchips in everyday household items, embedded devices, the availability of personal computers, and mobile devices.
Thus, our world is becoming more and more interconnected, creating vast networks of devices producing data at a large scale.


We have to process this constantly generated data in a way that satisfies consumer needs, such as zero downtime and low latency.
Algorithms over the incoming data stream cannot be computed locally on one central server due to the bitrate of the stream \cite{Lam2012}; furthermore, requests could come in from all over the world.
Therefore, using some distributed system is necessary to provide the power to meet these demands \cite{Leonard1985}, however partitioning our computation raises interesting challenges.

Partitioning often occurs with simple hash functions \cite{spark, nasir2015}, which does not take the weights of the keys into account.
This only works well when the distribution of the incoming data stream follows a near-uniform distribution \cite{nasir2015}.


However, requests over the internet do not follow a uniform distribution; instead, they follow Zip's law, and can, therefore, be fitted using a Zipfian distribution \cite{Adamic2002}.
This means that the majority of internet traffic may attribute to a small portion of the userbase.

Changing or non-uniform data distributions can cause overloaded partitions, and manually adjusting partitioners is infeasible due to the amount and frequency of incoming data \cite{Zliobaite2012}.
Unexpected and high loads in a distributed system can cause downtime, which entails significant losses in revenue.
Examples include Amazon's and Lowe's downtime during Black Friday \cite{business_insider_2018, cnn_business_2017}.

The base problem of partitioning is the NP-complete \textit{Bin Packing} problem's special case, the \textit{Multiprocessor Scheduling Problem} \cite{garey1979}.
Although it has a trivial polynomial-time solution, this would require tasks with the same lengths \cite{garey1979}.

For non-trivial cases, the problem requires the knowledge of every key's weight, which, in a distributed system, is a hard, if not nearly impossible task.

However, the problem can be alleviated by using the knowledge that the data follows Zip's law, therefore there exists a cutoff where the incoming requests are so sparse that it will not influence any partitioning algorithm significantly.
It becomes possible to count only a subset of the elements with the highest frequencies, called frequent elements problem.
Some of the earlier works researching the frequent elements problem were done by Alon \textit{et. al.} \cite{Alon1999} ; Henzinger \textit{et. al.} \cite{henzinger1998}; and Charikar \textit{et. al.} \cite{charikar2002}.

Algorithms that can precisely and cheaply estimate the frequencies of such elements are essential.
Stalling the data stream with heavy computations is infeasible, so solutions that can process each element in the stream in $O(1)$ time, query the most frequent items in sublinear time, and use sublinear space are well sought after \cite{Manku2002, metwally2005, demaine2002, misra1982, charikar2002, Cormode2005, Golab2003}.


A major problem in in data streaming is changing distribution, also called \textit{Concept drift}, which requires algorithms that can detect the changes without explicit knowledge about them \cite{Widmer1996}.
Failing to do so causes their accuracy to degrade \cite{Zliobaite2012}.
However overcorrection triggered by mild noise is also a problem, so balancing robustness and flexibility is key \cite{Stanley2003}.

Data bursting happens when data packets arrive at their destination more rapidly than intended by the transmitter \cite{allman2005}.
The phenomenon of micro-bursting is quite common in networks with window-based transport protocols, especially TCP \cite{allman2005}.

Distributed systems (and most parts of the internet) are built on TCP, but even on non-window based protocols with physically far enough components, where packets are going to be transmitted over several jumps, data burst can occur.
Algorithms that are not designed to handle periodical micro-bursting could show significant inaccuracies in their results.

In this paper, our main contributions are:
    a data generator that can handle \textit{concept drifts} and \textit{data bursts};
    a mechanism to benchmark and rank sampling algorithms, and a baseline \textit{Oracle}, which correctly estimates the ground truth based on our data generator;
    formalization of \textit{concepts} and \textit{concept drifts} and proof of the correctness of the \textit{Oracle};
    a hyperparameter optimizer for sampling algorithms;
    two novel sampling algorithms that can react to \textit{concept drifts} and analysis of them in conjunction with state-of-the-art sampling algorithms using our framework.

%% file: content/background.tex
\section{Related work}
\label{section:background}

\subsection{Sampling algorithms}

There are multiple algorithms that try to solve the problem of frequent items with fundamentally different approaches.
The common ground for all of these algorithms is that they have to carefully balance the memory usage, run time, and precision simultaneously.

These algorithms can be categorized into the following three categories: \textit{Counter-based}, \textit{Sketch-based} and \textit{Change respecting}.

\textit{Counter-based} algorithms rely on a counter to update the currently estimated frequencies of each key. Strategies exist to periodically flush or thin the counters to save memory.
Three defining algorithms of this category are:
\textit{Sticky Sampling} \cite{Manku2002}, \textit{Lossy Counting} \cite{Manku2002} and \textit{Space Saving} \cite{misra1982, demaine2002, metwally2005}.

\textit{Sketch-based} algorithms work by keeping fewer counters than their counter-based counterparts. Often using a probabilistic approach to keep estimations of the frequency in sub-linear space.
However they require expensive calculations when recording a key, so they cannot fit into the tight budget of a data streaming application.
Notable algorithms are: \textit{Count Sketch} \cite{charikar2002} and \textit{Count-Min Sketch} \cite{Cormode2005}.

\textit{Change respecting} algorithms have a mechanism in place to detect and conform to the concept drifts unlike static algorithms, which concentrate on either static data or a data stream with a non-changing distribution.
We are especially interested in this category, the two main algorithm here are \textit{Landmark} \cite{Zhu2002} and \textit{Frequent} \cite{Golab2003}.


\subsection{Partitioning algorithms}

Hash partitioning works by taking the modulus of the hash of the data key with the number of partitions \cite{spark}.
This is a static partitioning algorithm that only solves the \textit{Multiprocessor Scheduling problem} with uniformly distributed keys.

In distributed systems, the number of partitions can change and migration costs cannot be ignored, as parts of the system could be physically apart.
Consistent hash is therefore commonly used \cite{karger1997, Gedik2014}.
Its major problem is that it has difficulties with non-uniform distributions \cite{Gedik2014}.

Gedik proposed three algorithms that consider non-uniform distributions as well \cite{Gedik2014}.
Therefore, we decided to use Gedik's algorithms in our tests.
Furthermore, these algorithms are constructed in such a way that both balance and migration costs can be configured easily.

\subsection{Concept Drifts}

Widmer defines \textit{concept drift} as the radical change in a target concept introduced by changes in a hidden context \cite{Widmer1996}.
Studying \textit{concept drift} in data science is a fundamental building block to a good algorithm.
Internet traffic is ever-changing due to minor or major real-world events and drifts in the data will be inevitable.

A lot of research went into studying concept drifts, but the emphasis was placed on machine learning, with a categorizer on a finite number of categories.

The term \textit{concept drift} does not have a clear definition in the field of data streaming or it does not apply to the sampling problem.
It is often used interchangeably with the following terms: \textit{concept shift}, \textit{changes of classification}, \textit{fracture points}, \textit{fractures between data}, since these all refer to the same basic idea \cite{torres2012}.
Therefore, when we mention \textit{concept drift} or \textit{drifting data stream} we refer to the change of the underlying distribution of the stream.
For example, the cause of this change in the distribution could be a sporting event, a sale at an online store or sudden hardware failure at a major data centre.

\textit{Concept drift} can be sudden, also known as abrupt or instantaneous, where a change takes place suddenly; and gradual, where a transition period exists between two \textit{concepts} \cite{Stanley2003, tsymbal2004problem}.
Sudden \textit{concept drift} may occur because of critical failure at a major server park, while a gradual drift can occur every evening as different parts of the world start their day, while others go to sleep.
\textit{Concept drift} may also occur during any of our computations, which renders both static partitioning algorithms and static algorithms ineffective.

\subsection{MOA}

The closest framework to our proposed one is MOA \cite{moa}: 

It was built for AI development and its main strengths are built around this idea.
It has \textit{concept drift} generation capabilities, but it was designed with classification in mind.

For a machine learning setting over data streams, the authors of MOA formulate the following requirements \cite{moa}:
\begin{itemize}
\small{
    \item Process an example at a time, and inspect it at most once
    \item Use a limited amount of memory
    \item Work in a limited amount of time
    \item Be ready to predict at any time
}
\end{itemize}
These requirements are also integrated in our requirements.

MOA has no burst generation capabilities, which would allow it to make more real-world-like scenarios.
It also does not offer an out-of-the-box solution for metadata generation.
Metadata would allow us to make experiments repeatable and could be used to make accurate error calculations.
Concept drifts are only loosely defined in the paper \cite{moa}, a formal definition would allow us to craft a precise error calculation method and prove its correctness.

%% file: content/defs.tex
\section{Definitions}
\label{section:defs}

\subsection{Data streams}

We define \textit{data streams} with a finite cartesian product over the set of keys $K$. A key could be anything, we define $K$ over the natural numbers ($K \subset \mathbb{N}$).

\begin{definition}[Data Stream]
    Given a key-set $K$ and $n \in \mathbb{N}^+$, let $str \in K^n$ be a \textit{data stream}.
\end{definition}

Although it is not the nature of the stream to have a random accessor to their items, for further definitions let $str_i$ be the $i$-th element of the stream.

\subsection{Concepts}
\label{section:def_concepts}

\textit{Concepts} are probability distributions with a start and an end index, between which the \textit{concept} is considered active.

\begin{definition}[Concept]
    Let $Con \subset \mathbb{N}^+ \times \mathbb{N}^+ \times (\mathbb{N}^+ \rightarrow (K \rightarrow [0, 1]))$ be the set of all concepts, let $c = (c_s, c_e, c_f) \in Con$, if $c_s \leq c_e$ and $\forall i \in \mathcal{D}(c_f): c_f(i)$ is a discrete probability distribution over $K$, we call this $c$ a \textit{concept}.
\end{definition}

Given a stream $str \in K^n$ ($n \in \mathbb{N}^+$), let $Con^{str} \subset Con$ be the set of all \textit{concepts} for stream $str$.
We require exactly one \textit{concept} to be active at every point of the stream.
\[
   \forall i \in [1, 2, \dots, n], \exists! (c_s, c_e, c_f) \in Con^{str}: c_s \leq i \leq c_e
\]
    And require the \textit{concept}'s probability function to cover the range of the \textit{concept}.
\[
    \forall  (c_s, c_e, c_f) \in Con^{str}: \mathcal{D}(c_f) = [c_s, c_s + 1, \cdots, c_e]
\]

Given a stream $str \in K^n$ ($n \in \mathbb{N}^+$), $c = (c_s, c_e, c_f) \in Con^{str}$ and $\Delta \in ([0, 1] \times (K \rightarrow [0, 1]) \times (K \rightarrow [0, 1])) \rightarrow (K \rightarrow [0, 1])$ we study the following two types of $c_f$:
\begin{itemize}
    \item Constant \textit{concept}: $\exists P^c$ discrete probability distribution over $K$,
    $\forall i \in \mathcal{D}(c_f): c_f(i) = P^c$ 
    \item Changing \textit{concept}: if $c_e > c_s$ and $\exists P^c_1, P^c_2$ discrete probability distributions over $K$,
    $\forall i \in \mathcal{D}(c_f): c_f(i) = \Delta(\frac{i - c_s}{c_e - c_s}, P^c_1, P^c_2)$
\end{itemize}

Function $\Delta$ can be defined in multiple ways, we will provide further information on the specifics of how we use it in \autoref{section:generator}.

\begin{definition}[Concept Drift]
    Given a stream $str \in K^n$ ($n \in \mathbb{N}^+$),
    let $Con^{str} \subset Con$ be the set of all \textit{concepts} for stream $str$. \textit{concept drift} occurs in $i \in [1, 2, \dots, n-1]$ if,
\[
    \exists c = (c_s, c_e, c_f), c' = (c'_s, c'_e, c'_f): c_f(i) \neq c'_f(i+1)
\]
\end{definition}

\begin{definition}[Abrupt Concept Drift]
    Given two consecutive \textit{concepts}, $c = (c_s, c_e, c_f), c' = (c'_s, c'_e, c'_f), c_e + 1 = c'_s$, abrupt \textit{concept drifts} occur at $c'_s$ if $c_f(c_e) \neq c'_f(c'_s)$.
\end{definition}

Gradual \textit{concept drift} occurs between the start and end of a changing \textit{concept}.
\begin{definition}[Gradual Concept Drift]
    Given a stream $str \in K^n$ ($n \in \mathbb{N}^+$) and a changing \textit{concept} $c = (c_s, c_e, c_f) \in Con^{str}$ gradual \textit{concept drift} occurs during $[c_s, c_s + 1, \cdots , c_e]$
\end{definition}

\begin{figure}[H]
    \centering
    \includegraphics[width=0.7\linewidth]{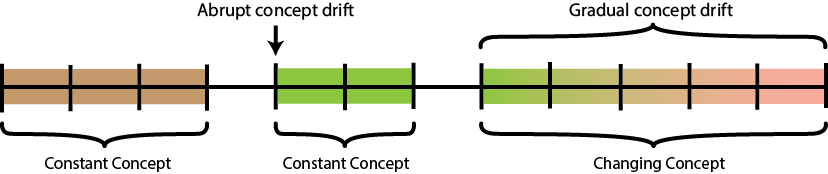}
    \caption{Concept drifts described by concepts}
    \label{fig:concepts_concept_drift}
\end{figure}

Certain combinations of $\Delta$ and range can be chosen for any changing \textit{concept} so that an abrupt drift could be defined.
We deem these as a misuse because a changing concept describes a gradual drift, which should happen over time, not suddenly.

Given all the \textit{concepts} that are acting on the stream, the true distribution can be defined at any location.

\begin{definition}[Concept and True Distribution]
\label{def:true_distr}
    Given a stream $str \in K^n$ ($n \in \mathbb{N}^+$), its \textit{concepts} $Con^{str} \subset Con$.
    At any point $i \in [1, 2, \dots, n]$ the true distribution of the stream $P^{i, str}_T$ is determined by the underlying \textit{concept}, $\exists! c = (c_s, c_e, c_f) \in Con^{str}): c_s \leq i \leq c_e$
\[
    P^{i, str}_T = c_f(i)
\]
\end{definition}

\subsection{Sampling algorithm}

Our problem is similar to the \textit{frequent items problem} \cite{Cormode2008}.

Beyond the most frequent items above a threshold, the frequency of those items is also necessary.

\begin{definition}[Sampling problem]
    Given a key-set $K$, $n \in \mathbb{N}^+$ and a \textit{data stream} $str \in K^n$, let the resulting most frequent items with their frequency be $F \in K \rightarrow \mathbb{N}^+$ with threshold $\phi \in [0\dots1]$ defined as
    \[
    \forall k \in K, (k, f_k) \in F \iff f_k \geq \phi \cdot n
    \]
    \[
    f_k= \sum_{i\in[1,2,\dots,n]} \mathbf{1}_k(str_i)
    \]
\end{definition}

Instead of $\phi$, a $k \in \mathbb{N}^+$ can be given, resulting in at most $k$ number of items with the highest frequencies.
This $k$ parameter is often called \textit{top-K} in the literature \cite{metwally2005}.

Algorithms that solve the sampling problem have to work with heavy limitations. Memory is limited, the run time has to be as low as possible, the stream can only be procesedded once, it's elements one-by-one and the length of the stream is hidden or non-existent.

The result of these algorithms are normalized to produce relative frequencies, we will call these \textit{sample distributions}.

\subsection{Oracle}

\begin{definition}[Oracle]
An oracle is a sampling algorithm that has access to the \textit{data stream}'s metadata, which contains the \textit{concepts} and the $\Delta$ function that was used to generate that \textit{data stream}. It can trivially calculate the true distribution ($P^{i, str}_T$) using the method described in definition \ref{def:true_distr}.
\end{definition}

Using the oracle as a sampling algorithm in a real distributed system is of course impossible.

\subsection{Error calculation}

\subsubsection{Distribution difference}
\label{subsubsection:distribution_difference}
Directly measuring the error of an algorithm in isolation is key to determine how fast and how accurately it reacts to drifts compared to other algorithms.
We need the baseline distribution, provided by the \textit{oracle sampling}, to which we can compare all of the other distributions.

There are multiple ways, to compare probability distributions, but usually, either Kullback–Leibler divergence or Hellinger distance is used \cite{fink2009, Webb2016}.

Given $P$ target discrete distribution and $Q$ approximate discrete distribution the Hellinger distance is the following formula \cite{Hellinger1909}:
\begin{equation}
    H(P \| Q) = \frac{1}{\sqrt{2}} \cdot \sqrt{\mathlarger{\sum}_{x} (\sqrt{P(x)} - \sqrt{Q(x)})^2}
\end{equation}

Hellinger distance is bounded between $0$ and $1$. It is a distance function (metric), and unlike Kullback–Leibler divergence, it does satisfy the rule of symmetry and triangle inequality.

\subsubsection{Load imbalance}

A partitioning problem can be traced back to the \textit{Multiprocessor scheduling problem} originally defined by Garey \textit{et. al.} in 1979 \cite{garey1979}. It is NP-complete, but pseudo-polynomial solutions exist for any fixed $m$. \cite{garey1979}.

When measuring the error of a whole system we use the most commonly used metric, the \textit{percent imbalance} metric. It measures the lost performance to the imbalanced load, or in other words the performance that could be gained by balancing our partitions \cite{pearce2012}.
Pearce \textit{et. al.} (2012) \cite{pearce2012} defined this imbalance as follows:

    \begin{definition}[Imbalance]
        Let $L$ be the loads on our partitions, then $L_{max}$ is the maximum load on any process. Let $\bar{L}$ be the mean load over all processes. The \textit{percent imbalance} metric for $L$, $\lambda_L$, is:
        
        \begin{equation}
            \lambda_L = \left( \frac{L_{max}}{\bar{L}} - 1 \right) \cdot 100%
        \end{equation}
    \end{definition}
    
    Calculating the $L_{max}$ of the loads is sufficient because that will determine the runtime of our whole computation.

%% file: content/frame/_frame.tex
\section{Our Framework}
\label{section:frame}
To test sampling algorithms we need a 
flexible, fast and deterministic
testbed
with low overhead.

\input{content/frame/requirements}

\input{content/frame/generator}

\input{content/frame/ranking}

\input{content/frame/optimizer}

%% file: content/frame/requirements.tex
We list our findings of what the requirements should be for a ranking and benchmarking framework for algorithms over data streams.
We base these requirements on previous frameworks \cite{moa} and our own experience while developing sampling algorithms.

\begin{itemize}
    \item Allow algorithms to only \textit{process the data stream at most once} \cite{moa}.
    \item Provide \textit{metrics} about the algorithms.
    \item Provide those metrics \textit{on-demand} during the computation \cite{moa}.
    \item Be \textit{Fast} with low overhead \cite{moa}.
    \item Be \textit{Deterministic}.
    \item Include a \textit{data generator} that can
    \begin{itemize}
        \item handle \textit{concept drifts};
        \item allow \textit{pre-generation} of \textit{data streams};
        \item generate the \textit{metadata} of pre-generated streams;
        \item be capable of simulating \textit{micro-bursting} during tests.
    \end{itemize}
    \item Offer \textit{Hyper-parameter optimization}.
\end{itemize}

%% file: content/frame/generator.tex
\subsection{Data generator}
\label{section:generator}

\subsubsection{Concept drift generation}
It makes more sense to define \textit{concepts} and then \textit{concept drifts} as their consequence.
However, for generating \textit{concept drifts} it is more straightforward to define the points at which \textit{concept drifts} are occurring and derive the \textit{concepts} from those.

Let $D \subset \mathbb{N}^+_0 \times \mathbb{Q}^+ \times (K \rightarrow [0,1]) \times (K \rightarrow [0,1])$, $(len, mid, P_1, P_2) \in D$ is a \textit{drift} if
($mid \in \mathbb{N}^+ \land len \mid 2$) $\lor$ ($mid + \frac{1}{2} \in \mathbb{N}^+ \land len \nmid 2$), where $len$ is the length, $mid$ is the midpoint of the drift and $P_1, P_2$ are the two probability distributions of the two \textit{concepts} that the \textit{drift} is created by.

In further definitions $\forall d = (len, mid, P_1, P_2) \in D$ let $s \in D \rightarrow \mathbb{N}^+: s(d) = mid - \frac{len}{2}$ and $e \in D \rightarrow \mathbb{N}^+: e(d) = mid + \frac{len}{2}$.

We have exactly one concept active at any given point, so overlapping \textit{drifts} should not be allowed either.
Given $str \in K^n (n \in \mathbb{N}^+)$ and $D^{str} \subset D$, the \textit{drifts} for that \textit{data stream}: $\forall d, d' \in D^{str}, d \neq d': $
\[
    e(d) \leq e(d') \implies
    e(d) < s(d')
\]

There also has to be an initial \textit{drift} starting at index $1$, 
\[
    \exists d \in D^{str}: s(d) = 1
\]

For a $str \in K^n$, ($n \in \mathbb{N}^+$) stream we call $d = (len, mid, P_1, P_2) \in D$ a gradual \textit{drift}, if $len > 0$, $s(d) >= 1$ and $e(d) <= n$.

For a $str \in K^n$, ($n \in \mathbb{N}^+$) stream we call $d = (0, mid, P_1, P_2) \in D$ an abrupt \textit{drift}, if $s(d) >= 1$ and $e(d) <= n$.

Given a stream $str \in K^n (n \in \mathbb{N}^+)$ generated by $D^{str} \subset D$ drifts, $\forall i \in [1, 2, \dots, n]$, let $str_i$ be a random variable following $P^{i, str}_G$ discrete distribution, where $P^{i, str}_G$ is
\begin{enumerate}
    \item $\exists d = (d_{len}, d_{mid}, d_{P_1}, d_{P_2}) \in D^{str}: len > 0 \land s(d) \leq i \leq e(d)$ then,
    \[
        P^{i, str}_G := \Delta(\frac{(i - s(d))}{d_{len}}, d_{P_1}, d_{P_2})
    \]
    \item $\exists d = (0, d_{mid}, d_{P_1}, d_{P_2}) \in D^{str}: s(d) = i$ then,
    \[
        P^{i, str}_G := d_{P_2}
    \]
    \item $\exists d = (l, m, d_{P_1}, d_{P_2}) \in D^{str}:$ 
    $e(d) < i
        \land
    \nexists d' \in D^{str}: e(d) < s(d') \leq i$
    \[
        P^{i, str}_G := d_{P_2}
    \]
\end{enumerate}

To make sure that these drift definitions are compliant with the \textit{concepts} defined in \autoref{section:def_concepts}, we prove that their expressive powers are equal (\autoref{thm:dr_co_eq_I}, \autoref{thm:dr_co_eq_II}).
Therefore the same \textit{concept drifts} can be described by them.

\begin{theorem}[Drifts and concepts are equal in expressive power I.]
    \label{thm:dr_co_eq_I}
    $\forall str \in K^n$ and $Con^{str} \subset Con$: $\exists D^{str} \subset D$ drifts,
    where $|D^{str}| = |Con^{str}|$ and
    $\forall i \in [1,2, \dots, n]: P_G^{i, str} = P_T^{i, str}$
\end{theorem}
\begin{proof}
    We construct such $D^{str}$, $\forall c = (c_s, c_e, c_f) \in Con^{str}$ and show that the construction is correct.
    Then for any $i \in [1, 2, \cdots, n], \exists! c = (c_s, c_e, c_f) \in Con^{str}: c_s \leq i \leq c_e$ (based on requirements for $Con^{str}$). We show that for any $c$ concept, our constructed corresponding $d$ drift is correct by showing that $P^{i, str}_G = P^{i, str}_T$.
    Please see \autoref{appendix:section:equal_I} for the whole proof.
\end{proof}

\begin{theorem}[Drifts and concepts are equal in expressive power II.]
    \label{thm:dr_co_eq_II}
    $\forall str \in K^n$ and $D^{str} \subset D$: $\exists Con^{str} \subset Con$ concepts,
    where $\forall i \in [1,2, \dots, n]: P_T^{i, str} =  P_G^{i, str}$
\end{theorem}
\begin{proof}
    We begin by constructing such $Con^{str}$, $\forall d = (d_{len}, d_{mid}, d_{P_1}, d_{P_2}) \in D^{str}$,
    then for any $i \in [1, 2, \cdots, n]$, we show that for any $d$ drift, our constructed corresponding $c$ concept is correct by showing that $P^{i, str}_T = P^{i, str}_G$.
    Please see \autoref{appendix:section:equal_II} for the whole proof.
\end{proof}

Both the generator and the concepts depend on a $\Delta$ function.
To allow the most generic way of defining drifts we left $\Delta$ undefined up until this point.
For performance reasons, we only approximate a linear interpolation with the following $\Delta$ function, and in \autoref{thm:oracle_correct} we prove the correctness of this approximation.

Let $\Delta \in ([0, 1] \times (K \rightarrow [0, 1]) \times (K \rightarrow [0, 1])) \rightarrow (K \rightarrow [0, 1]): \forall p \in [0, 1]$
\[
    \multifunction{\Delta(p, P_1, P_2)}
    {P_1}
    {\text{if } R > p}{P_2}
\]
where $R$ is a random variable following standard uniform distribution.

We use this $R$ to avoid recalculating a new distribution for every generated item.

To see smooth transitions in the generator's implementation we also require that the probability distribution at the final point of a drift matches the probability distribution at the first point of the following drift.
To achieve this we require the drifts to comply with the following restriction:

Let $str \in K^n$ ($n \in \mathbb{N}^+$) be a \textit{data stream} and $D^{str}$ its concepts.
\[
    \forall d = (d_{len}, d_{mid}, d_{P_1}, d_{P_2}) \in D^{str}:
\]
\begin{gather*}
    \exists d' = (d'_{len}, d'_{mid}, d'_{P_1}, d'_{P_2}) \in D^{str}, \\ d'_{mid} > d_{mid}
    \implies
    d_{P_2} \equiv d'_{P_1}
\end{gather*}

Using the aforementioned $\Delta$ it can be shown that this is sufficient precondition.

\subsubsection{Dynamic burst generation}

When simulating bursts we simulate faulty routers.
Before beginning every \textit{micro-batch}, a burst has a chance to start.
During the duration of the burst, the keys have a certain probability that they will be held back.
At the end of the burst, which could last a couple of micro-batches, the keys are released back into the stream at once.

Bursts can only be applied to already existing streams, for this, it is a prerequisite for a burst process to already have a \textit{data stream} prepared.
We will use the phrase \textit{loading} to describe the process of getting the next item of this stream.

Bursts are defined with the following structure: $B \subset [0, 1] \times [0, 1] \times \mathbb{N}^+ \times \mathbb{N}^+$. $b = (bsp, kbp, bl_{min}, bl_{max}) \in B$ is a valid burst configuration, if $bwl_{min} \leq bwl_{max}$.

\textit{Burst start probability} ($bsp$) describes the probability of a burst starting at any micro-batch.
Once the burst started, another one cannot start again.
Before any bursting could take place, the faulty keys are calculated in advance, with \textit{Key burstiness probability} ($kbp$) that gives the probability of each key being held up.
This process generates a map (\textit{Faulty Keys Map} - $FKM$) in which we can store these keys until the burst is over.

During the burst, the faulty keys are counted in $FKM$.
Whenever a key that is in $FKM$ would be \textit{loaded} into the stream as a next item, instead will be put into the $FKM$.
If a key is not faulty, then the stream loading can continue as usual.

The burst duration is defined by \textit{Burst length} in micro-batches, which is a random number between the minimum ($bl_{min}$) and maximum ($bl_{max}$) duration.
At the end of the burst, the held-up data in $FKM$ is released back into the stream.

%% file: content/frame/ranking.tex
\subsection{Ranking}
\label{section:ranking}

We use on-demand metric querying to archive the ranking of sampling algorithms.
We show two different ways to compare them.
One method works with the direct sampling outputs, while the other method indirectly measures them in a simulated environment by the efficiency of that simulated system.

\subsubsection{Ranking by reported distribution difference}
Sampling algorithms can be measured directly by the samples described in \autoref{subsubsection:distribution_difference}.

To achieve this, we run the algorithms in separate processes and provide them the same \textit{data stream}.
At the end of each micro-batch, the samples of the algorithms are saved and the sample distributions are calculated from them.
We run an \textit{Oracle} on the same \textit{data stream} to provide the baseline, the true distribution, at the end of each micro-batch.

The \textit{Oracle} is a good choice for this, because it can correctly calculate the true distributions even during concept drifts (\autoref{thm:oracle_correct}).
Seeing the difference increase between a sample distribution and the true distribution during concept drifts means the sampling algorithm could not react to the drift fast enough.
After this, the speed at which the difference shrinks tells us how fast it can detect and correct the drift.
The \textit{Oracle}'s correctness is based on the synthetic data generator.
We showed that the data generator can correctly generate concept drifts in \autoref{thm:dr_co_eq_I} and \autoref{thm:dr_co_eq_II}.
Therefore, it is enough to show that the oracle correctly calculates the true distribution based on the data generator.
We have access to all of the drifts $D^{str} \subset D$ from the metadata, which was used to generate $str \in K^n$ \textit{data stream}.

Because there is always at least one drift present (the starting drift) and no two drifts can overlap at any given point, for all $i \in [1, n]$, $P^{i, str}_O$ is:
\begin{enumerate}
    \item $\exists d = (d_{len}, d_{mid}, d_{P_1}, d_{P_2}) \in D^{str}: d_{len} > 0 \land s(d) \leq i \leq e(d)$ then,
    \[
        P^{i, str}_O := (1-\frac{(i - s(d))}{d_{len}}) \cdot d_{P_1} + \frac{(i - s(d))}{d_{len}} \cdot d_{P_2}
    \]
    \item $\exists d = (0, d_{mid}, d_{P_1}, d_{P_2}) \in D^{str}: s(d) = i$ then,
    \[
        P^{i, str}_O := d_{P_2}
    \]
    \item  $\exists d = (l, m, d_{P_1}, d_{P_2}) \in D^{str}: e(d) < i
        \land
    \nexists d' \in D^{str}: e(d) < s(d') \leq i$
    \[
        P^{i, str}_O := d_{P_2}
    \]
\end{enumerate}

\begin{theorem}[Oracle is correct]
    \label{thm:oracle_correct}
    Given $str \in K^n$ ($n \in \mathbb{N}^+$) stream generated by $D^{str} \subset D$ drifts.
\[
    \forall i \in [1, 2, \dots, n] : P^{i, str}_O \textit{ correctly estimates } P^{i, str}_G
\]
\end{theorem}
\begin{proof}
There are two cases for each $i \in [1, 2, \dots, n]$:
\begin{enumerate}
    \item $\exists d = (d_{len}, d_{mid}, d_{P_1}, d_{P_2}) \in D^{str}: d_{len} > 0 \land s(d) \leq i \leq e(d)$
    
    let $p = \frac{(i - s(d))}{d_{len}}$, the generator defines $P^{i, str}_G$ to be
    \[
        P^{i, str}_G := \Delta(\frac{(i - s(d))}{d_{len}}, d_{P_1}, d_{P_2})
    \]
    we provided the generator with the following $\Delta$ function,
    \[
        \multifunction{\Delta(p, d_{P_1}, d_{P_2})}
        {d_{P_1}}
        {\text{if } R > p}{d_{P_2}}
    \]
    where $R$ is a random variable following standard uniform distribution.
    \begin{equation*}
        \begin{split}
        P(P^{i, str}_G & = d_{P_1}) = P(R > p) = (1 - p) \\ 
        &= d_{P_2}) = P(R \leq p) = p
        \end{split}
    \end{equation*}
    \begin{gather*}
    \end{gather*}
    therefore, the expected value of $P^{i, str}_G$,
    \[
        E[P^{i, str}_G] = (1 - p) \cdot d_{P_1} + p \cdot d_{P_2} = P^{i, str}_O
    \]
    \item otherwise, $P^{i, str}_O \equiv P^{i, str}_G$ by definition.
\end{enumerate}

\end{proof}

\subsubsection{Ranking by load imbalance}

\begin{figure}[H]
    \centering
    \includegraphics[width=0.7\linewidth]{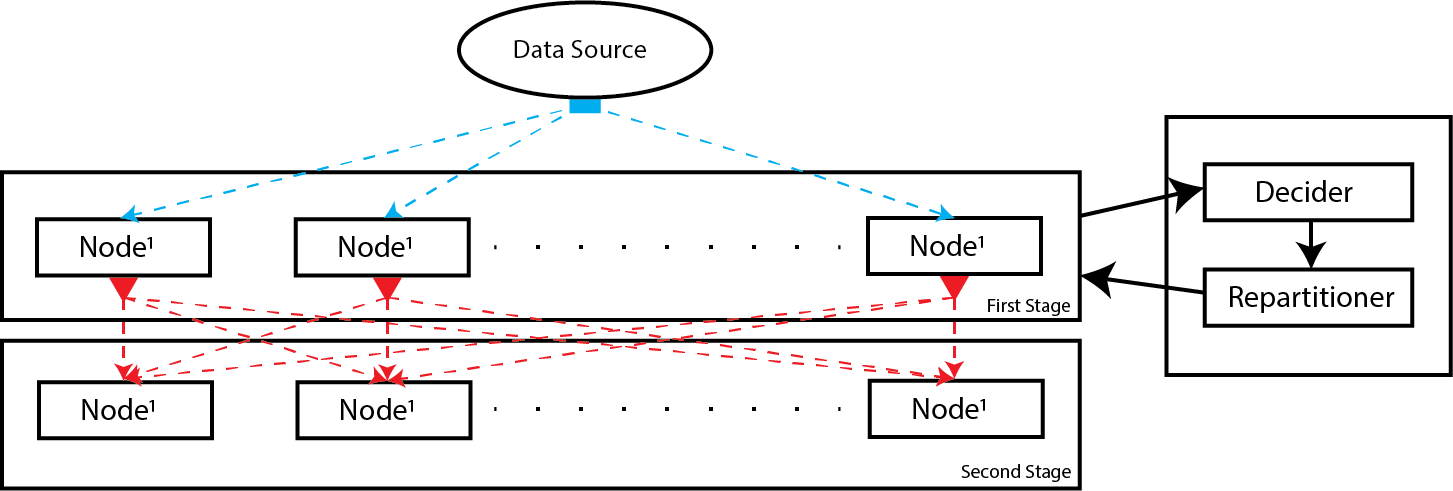}
    \caption{Flow of data in a distributed computation}
    \label{fig:sys_dig_rep_en}
\end{figure}
 
In real distributed systems incoming data is processed by multiple nodes.
The data arrives at the \textit{first stage} from a data source, such as Kafka, Couchbase, through a hash partitioner denoted by the blue square.
The nodes in the \textit{first stage} then shuffle this data by a grouping, which is a frequently used operation and could be the result of group-by or join operations.
The \textit{second stage} runs calculations on the shuffled \textit{data stream}.

In our framework, we simulate a distributed system with multiple nodes by starting multiple sampling algorithms.
We process the \textit{data stream} in micro-batches and for each element of the micro-batch, we select a sampling algorithm with a hash partitioner to feed that element to.
After each micro-batch, we gather the outputs of the sampling processes and aggregate them.
Based on the aggregated output, we determine whether repartitioning is necessary with a decider strategy.
If so, we create a new partitioning algorithm for the imbalance measurement, which will be used by the next micro-batch.

The calculation of the load imbalance happens during each micro-batch.
A naive approach would be to calculate the partition loads from the reported data of the sampling processes, however, this would allow cheating, the sampling algorithm could report falsified data and produce perfect results.
To make sure this does not happen, we calculate the partition loads during the micro-batch processing.
We determine this load by counting the actual number of elements that would be shuffled to certain partitions if we were to continue the computation.
If repartitioning happens after processing a micro-batch, the load imbalance caused by the new partitioning algorithm will be calculated during the following micro-batch.

To rank multiple sampling algorithms, we only compare test cases that have the same \textit{data stream}, number of nodes, decider strategy, repartitioner, and micro-batch size.

%% file: content/frame/optimizer.tex
\subsection{Optimizer}
\label{section:opt}

We provide a framework for optimization of sampling algorithms in which a new optimization strategy is easy to implement.

The optimization process works the following way:
\begin{enumerate}
    \item The initial population is selected as a starting point.
    \item The initial population is benchmarked 
    \item In the selection process, the \textit{selectors} are used to thin out the population based on their fitness.
    \item In the evolution phase, a new population is generated based on the previous generation's survivors.
    \item The survivors are added to the new population 
\end{enumerate}

\begin{figure}[H]
    \centering
    \includegraphics[width=0.5\linewidth]{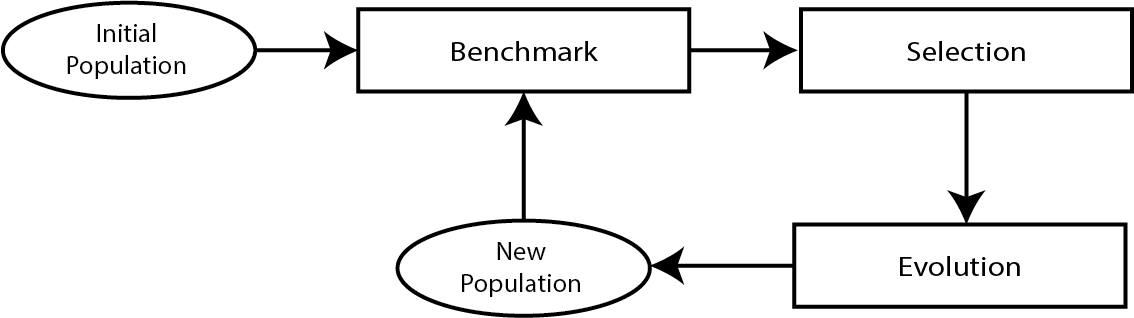}
    \caption{State diagram of our optimization process.}
\end{figure}

The optimization strategy we provide is a local minimum search, in which we choose the member of the population with the best fitness value in the selection process.

The hyperspace of an algorithm could be high dimensional.
Every additional parameter increases the number of neighbours exponentially.
Therefore, calculating all possible neighbours of an algorithm is infeasible and instead, we generate a subset of them and only include those in the population.
Because of the complexity of the problem domain, the correct steps and ranges of the parameters are not universal.
For example, a ''probability'' parameter only makes sense, if it is between $0$ and $1$.
Furthermore, it often does not make sense to use the whole domain of the parameter as the range.
For example, a ''window size'' parameter should not be bigger than the stream length and should have a sensible minimum size as well.

To evaluate these algorithms with specific hyperparameters, fitness functions are used by the optimizer.
They can be defined by the use of the metrics, such as accuracy, run time and memory usage.

If the diversity of the data stream is low, overfitting could occur during the optimization process.
To avoid this multiple \textit{data streams} are used during optimization, and also allow different kinds of bursts to be introduced to the \textit{data streams}.

%% file: content/solution.tex
\section{Our Algorithms}
\label{section:solution}

\subsection{Temporal Smoothed}

Our first algorithm is called \textit{Temporal Smoothed}. This algorithm is inspired by the \textit{Landmark}, and it aims to solve a flaw in it. In the \textit{Landmark}, when the sampler is asked to report shortly after a landmark, the reported output is too small in size and may not show an accurate distribution in its sampled keys.

Our algorithm, like the \textit{Landmark}, works with windows (\textit{threshold}), but rather than resetting the whole inner state of the sampling after each window, it starts a new sampling process instead. It maintains the original sampler (\textit{main sampler}) and the newly created one (\textit{secondary sampler}) for a predetermined window size (\textit{switch threshold}). During this window, the incoming data is sampled by both of them and only the main sampler's results are being reported. After the window is processed, the \textit{secondary sampler} gains a stable size and becomes the new \textit{main sampler}.

\SetKwInOut{Parameter}{Parameter}
\SetKwInOut{Input}{Input}
\SetKwInOut{Output}{Output}

\begin{algorithm}
  \footnotesize
  \Input{key: Key}
  \Parameter{$ms$: SamplerBase, \textit{main sampler}}
  \Parameter{$ss$: SamplerBase, \textit{secondary sampler}}
  \Parameter{$t$: $\mathbb{N}^+$, \textit{threshold}}
  \Parameter{$st$: $N^+$, \textit{switch threshold}}
  \BlankLine
    $totalProcessedElements \longleftarrow totalProcessedElements + 1$\;
    ms.recordKey(key)\;
    \If{ss initialized}{
        ss.recordKey(key)\;
    }
    \uIf{ss not initialized and $ms.totalProcessedElements = t + st$}{
        initialize ss\;
    }
    \ElseIf{ss initialized and $ss.totalProcessedElements = st$}{
        $ms \longleftarrow ss$\;
        discard ss\;
    }
  \caption{TemporalSmoothed}
\end{algorithm}

The \textit{TemporalSmoothed} algorithm can encapsulate any sampling algorithm.
As a result, the memory usage and run time heavily depend on which algorithm is encapsulated.

\begin{theorem}[Memory usage of \textit{Temporal Smoothed}]
    Assuming the chosen sampling algorithm has a maximum memory usage of $f(n)$ where $n$ is the current length of the stream, the \textit{TemporalSmoothed} algorithm's maximum memory usage will be:
    $$f(2 \cdot st + t) + f(st)$$
\end{theorem}
\begin{proof}
    First, it has to be stated the $f$ function is monotone increasing, since the maximum memory usage of a stream cannot decrease with the increase of the stream length.
    Secondly, with the help of \autoref{fig:tmp_lem_1_proof}, it can be shown that the \textit{main sampler} during its lifetime will sample $2 \cdot st + t$ times, which means that the maximum memory usage for the \textit{main sampler} will be $f(2 \cdot st + t)$.
    During the \textit{main sampler}'s lifetime, the \textit{secondary sampler} samples $st$ times, which means a maximum memory usage of $f(st)$.
\end{proof}

\begin{figure}[H]
    \centering
    \includegraphics[width=0.5\linewidth]{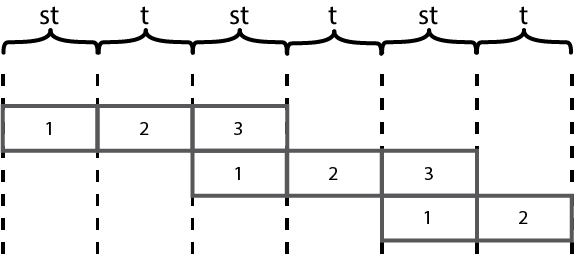}
    \caption{\textit{Temporal Smoothed} phases}
    \label{fig:tmp_lem_1_proof}
\end{figure}

Assuming the chosen sampling algorithm can initialize itself and sample a key in $\mathcal{O}(1)$ time, then the \textit{Temporal Smoothed} can also sample a key in $\mathcal{O}(1)$ time.
This assumption is not unreasonable, the sampling process has to sample large amounts of data in real-time.

\begin{theorem}[\textit{Temporal Smoothed} drift detection]
    \label{lem:tmp_lem_1}
    The \textit{Temporal Smoothed} algorithm can detect a change in the stream after sampling less then $t + 2 \cdot st$ number of times.
\end{theorem}
\begin{proof}
    In the algorithm a newly created sampler has three phases (\autoref{fig:tmp_lem_1_proof}).
    
    \begin{itemize}
        \item In phase one, the sampler takes the role of the \textit{secondary sampler}. The length of this phase is $st$ samples. (The exception to this is the very first sampler which already starts as \textit{main sampler})
        \item In phase two, the sampler becomes a \textit{main sampler}. The length of this phase is $t$ samples.
        \item In phase three, the sampler is still the \textit{main sampler}, but at the end of this phase, it gets discarded and replaced. This phase overlaps the newly created \textit{secondary sampler}'s phase one. The length of this phase is $st$ samples.
    \end{itemize}
    
    By examining the phase in which the change in the stream happens, we can determine the following cases:
     \begin{enumerate}
        \item The change happens exactly at the start of phase one. The change will be detected when the sampler becomes the \textit{main sampler}, which is in $st$ samples. (This cannot happen for the very first sampler, because the starting point of the change and the stream would overlap)
        \item The change happens in phase one or two. The sampler may not correctly detect the change since it has samples from before the change and after it. On the other hand, the next \textit{main sampler} will only have samples from after the change and can correctly detect it. This will happen in no more than $t + 2 \cdot st$ samples, which is the upper estimate for the replacement of the current \textit{main sampler}.
        \item The change happens in phase three. This case is covered by the first two cases because this coincides with the phase one of the current \textit{secondary sampler}.
    \end{enumerate}
\end{proof}

\subsection{Checkpoint Smoothed}

Our next algorithm is called \textit{Checkpoint Smoothed}. This algorithm works similarly to \textit{Temporal Smoothed}, but rather than periodically and rigidly renewing the \textit{main sampler}, it aims to only replace it if we are certain enough that the \textit{main sampler}'s reported results are not accurate due to a possible change in the stream. 

Our algorithm works with a \textit{main sampler} and a \textit{secondary sampler}. After our main sampler samples enough data (\textit{checkpoint window}), a new \textit{secondary sampler} is initialized.
The two samplers sample concurrently for another window (\textit{check threshold}).
After this, the reported results of the \textit{main sampler} and \textit{secondary sampler} gets compared using \textit{Hellinger distance}.
If the result is beyond a predetermined threshold (\textit{error threshold}), the \textit{main sampler} is replaced by the \textit{secondary sampler}.
We repeat the aforementioned process, either with the original \textit{main sampler} or the new one.

\begin{algorithm}
  \footnotesize
  \Input{key: Key}
  \Parameter{$ms$: SamplerBase, \textit{main sampler}}
  \Parameter{$ss$: SamplerBase, \textit{secondary sampler}}
  \Parameter{$cw$: $\mathbb{N}^+$, \textit{checkpoint window}}
  \Parameter{$ct$: $\mathbb{N}^+$, \textit{check threshold}}
  \Parameter{$et$: $R^+$, \textit{error threshold}}
  \Parameter{$checkpoint$: $N^+$, helper variable}
  \BlankLine
    $totalProcessedElements \longleftarrow totalProcessedElements + 1$\;
    ms.recordKey(key)\;
    \If{ss initialized}{
        ss.recordKey(key)\;
    }
    \uIf{ss not initialized and $totalProcessedElements > checkpoint + cw$}{
        initialize ss\;
    }
    \ElseIf{ss initialized and $ss.totalProcessedElements > ct$}{
        $expected \longleftarrow ms.estimateRelativeFrequencies$\;
        $actual \longleftarrow ss.estimateRelativeFrequencies$\;
        $result \longleftarrow H(expected, actual)$\;
        \If{$result > et$}{
        $ms \longleftarrow ss$\;
        }
        discard ss\;
        $checkpoint \longleftarrow totalProcessedElements$\;
    }
  \caption{CheckpointSmoothed}
\end{algorithm}

The \textit{Checkpoint Smoothed} algorithm can also encapsulate any sampling algorithm.

\begin{theorem}[Memory usage of \textit{Checkpoint Smoothed}]
Assuming the chosen sampling algorithm has a maximum memory usage of $f(n)$ where $n$ is the current length of the stream, the \textit{Checkpoint Smoothed} algorithm's maximum memory usage will be:
$$f(n) + f(ct)$$
\end{theorem}
\begin{proof}
    If there is too little or no change in the stream, the \textit{main sampler} will never be replaced.
    This means that the \textit{main sampler}'s maximum memory usage is $f(n)$.
    If the \textit{main sampler} never gets replaced, a new \textit{secondary sampler} is started periodically and will sample $ct$ amount of data.
    This means that the maximum memory usage of the \textit{secondary sampler} is $f(ct)$.
\end{proof}

The sampling time of a key for the \textit{Checkpoint Smoothed} algorithm depends on multiple things:
\begin{enumerate}
    \item The chosen sampling algorithm: As stated in the \textit{Temporal Smoothed} algorithms run time, it is not unreasonable to assume that a sampler can initialize and sample a key in $\mathcal{O}(1)$.
    \item To calculate the \textit{Hellinger distance} periodically, the relative frequencies of the \textit{main sampler} and \textit{secondary sampler} are needed.
    The run time is therefore tied to the size of the \textit{main sampler}'s and the \textit{secondary sampler}'s output.
    \item The calculation of the \textit{Hellinger distance} can be done in linear time based on the size of its inputs, therefore this does not increase the asymptotic run time.
\end{enumerate}
Based on these, the \textit{Checkpoint Smoothed} algorithm can sample a key in $\mathcal{O}(g(n) + g(ct))$ time.
Where $g(x)$ is the maximum number of relative frequencies calculated by the chosen sampling algorithm, given $x$ samples.

The asymptotic run time can be quite misleading, because the majority of times ($cw + ct - 1$ out of $cw + ct$ times) the run time will be $\mathcal{O}(1)$.
The algorithm's sample time can be improved upon.
For example, if we run the calculation of \textit{Hellinger distance} concurrently, we do not have to interrupt the sampling process.
This introduces a possible error, which is that we may calculate the errors for a \textit{main sampler} with $cw + \epsilon_1$ samples and for a \textit{secondary sampler} with $ct + \epsilon_2$ samples.
If the parameter $ct$ and $cw$ are reasonable in size this will not cause a significant change in the algorithm's output.

%% file: content/anal.tex
\section{Analysis}
\label{section:anal}

We use \textit{concept drifts} where the midpoint of the \textit{concept drift} is at the midpoint of the \textit{data stream}.
This is to give the algorithms as much time to react to the \textit{concept drifts} as they had to estimate the frequencies before the \textit{concept drift}.

All data sets consist of $5\,000\,000$ elements with a key set of size $100\,000$.
In our tests, we use micro-batches of size $30\,000$ and a \textit{top-K} value of $300$, because with the chosen Zipfian distributions this \textit{top-K} should include the keys with meaningful frequencies ($> 1\%$).

The number of test cases grows exponentially with each new algorithm and hyperparameter. Showing all of these would be impractical, because we have limited space, so instead we present only the relevant cases and provide all of our measurements on GitHub \footnotemark.
The measurements were made on a $6$ core, $12$ thread Ryzen $3600$ CPU clocked at $3.6$GHz with $16$GB of RAM.

\footnotetext{\url{https://github.com/g-jozsef/sampling-framework-aux}}

We only show our algorithms in this paper with \textit{Frequent} and \textit{Lossy Counting}.
\textit{Frequent} performed the best in the majority of our tests and only seems to have difficulties with bursts.
We use \textit{Lossy Counting} as the encapsulated sampling algorithm in \textit{Temporal Smoothed} and \textit{Checkpoint Smoothed}, due to its speed, low memory usage and accuracy.
\textit{Lossy Counting} is robust, it is resilient against data burst and is amongst the worst performers when it comes to reacting to \textit{concept drifts}.

We use the same parameters for \textit{Temporal Smoothed} and \textit{Checkpoint Smoothed} as for \textit{Landmark} to allow a fair comparison.
After running multiple optimizations over various datasets, we decided to use an error threshold of $0.2$ for \textit{Checkpoint Smoothed}.
\begin{itemize}
    \item \textit{Temporal Smoothed}: \textit{threshold}$ = 40000$, \textit{switch threshold}$ = 40000$
    \item \textit{Checkpoint Smoothed}: \textit{checkpoint window}$ = 40000$, \textit{check threshold}$ = 40000$, \textit{error threshold}$ = 0.2$
\end{itemize}

\begin{figure*}
  \centering
  \subfloat[gradual, $exp = 1$]{\includegraphics[width=0.33\linewidth]{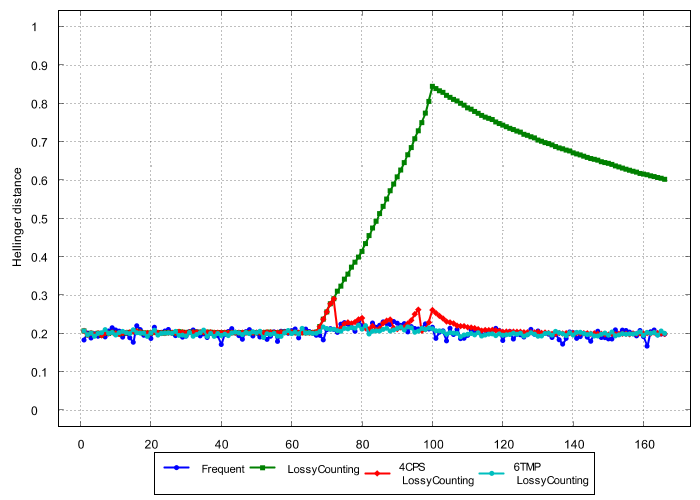}}
  \hfill
  \subfloat[gradual, $exp = 2$]{\includegraphics[width=0.33\linewidth]{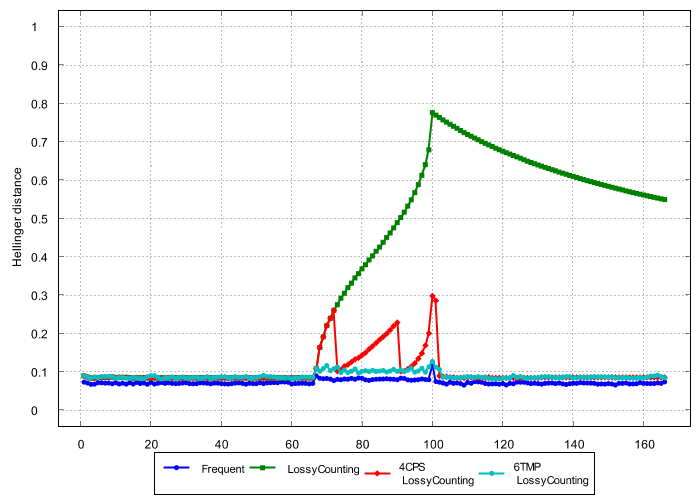}}
  \hfill
  \subfloat[gradual, $exp = 1$, light burst]{\includegraphics[width=0.33\linewidth]{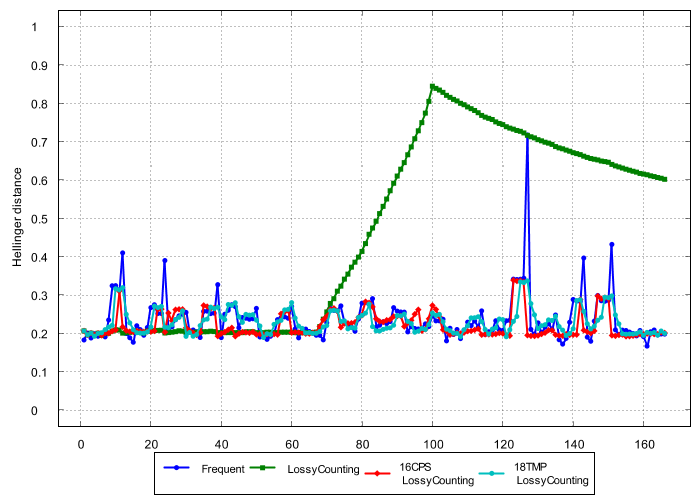}}
  \vfill
  \subfloat[gradual, $exp = 2$, light burst]{\includegraphics[width=0.33\linewidth]{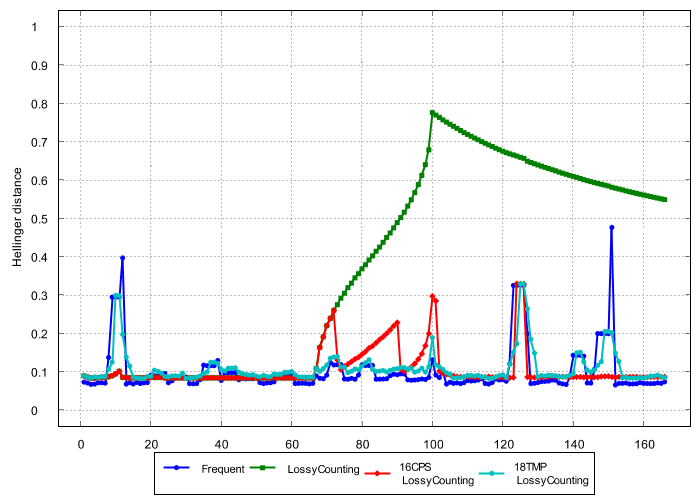}}
  \hfill
  \subfloat[gradual, $exp = 1$, heavy burst]{\includegraphics[width=0.33\linewidth]{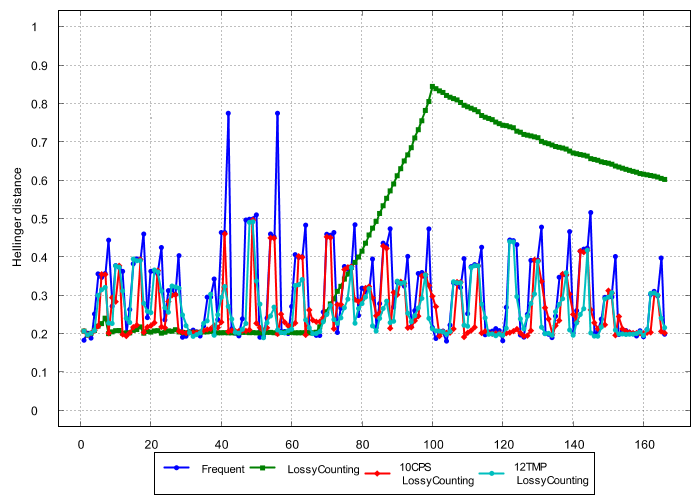}}
  \hfill
  \subfloat[gradual, $exp = 2$, heavy burst]{\includegraphics[width=0.33\linewidth]{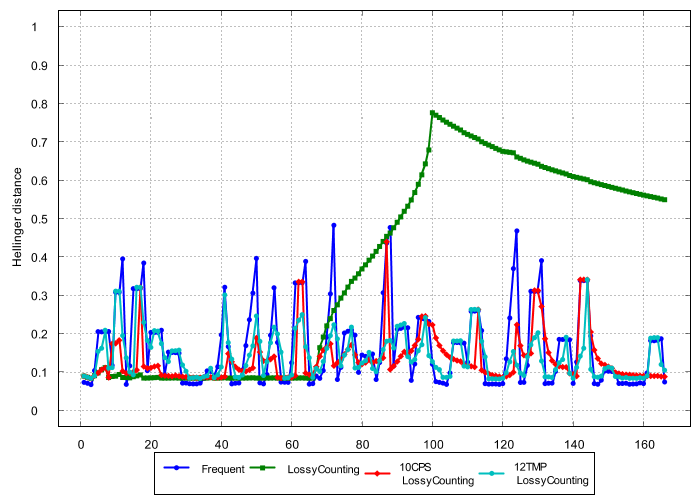}}
  \caption{\textit{Hellinger distance} between the sample distribution and true distribution at each micro-batch for \textit{Frequent}, \textit{Lossy Counting}, \textit{Checkpoint Smoothed} (\textit{CPS}) and \textit{Temporal Smoothed} (\textit{TMP}) without burst, and with light and heavy burst introduced.}
\end{figure*}

\textit{Temporal Smoothed} and \textit{Frequent} react the fastest to concept drifts, which is due to their windowed nature.
\textit{Checkpoint Smoothed} can be seen slightly inaccurate during the \textit{concept drift}, but it quickly corrects itself.
During the gradual drift, we can see multiple cuts made by \textit{Checkpoint Smoothed}, but after the concept drift it rapidly corrects itself.

With the introduction of light burst \textit{Checkpoint Smoothed} gives the best results on average because it is not hypersensitive to tiny changes.
\textit{Temporal Smoothed} and \textit{Frequent} however give better results between light bursts compared to \textit{Checkpoint Smoothed}.

With the introduction of heavy burst \textit{Checkpoint Smoothed} can better dampen the effects of bursts compared to \textit{Frequent} and \textit{Temporal Smoothed}.

\begin{figure*}
  \centering
  \subfloat[gradual, $exp = 1$]{\includegraphics[width=0.23\linewidth]{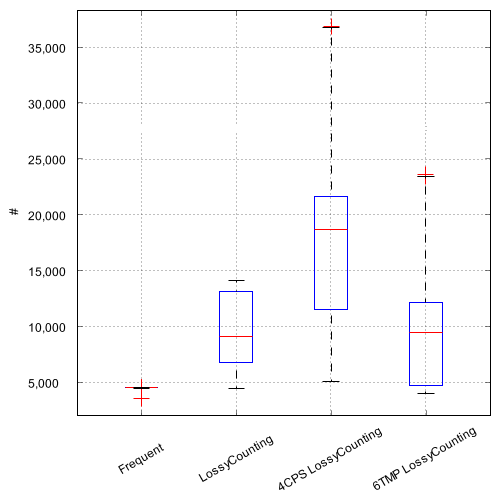}}
  \hfill
  \subfloat[gradual, $exp = 1$, heavy burst]{\includegraphics[width=0.23\linewidth]{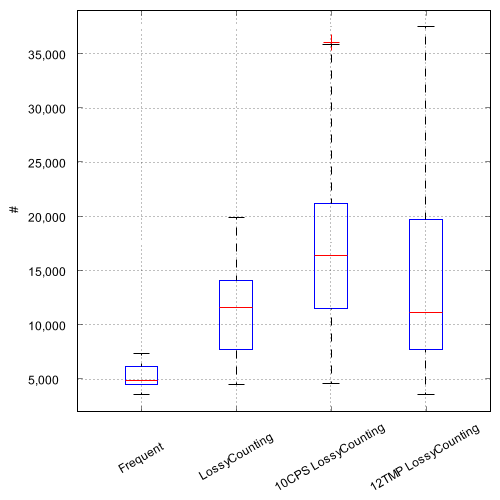}}
  \hfill
  \subfloat[gradual, $exp = 2$]{\includegraphics[width=0.23\linewidth]{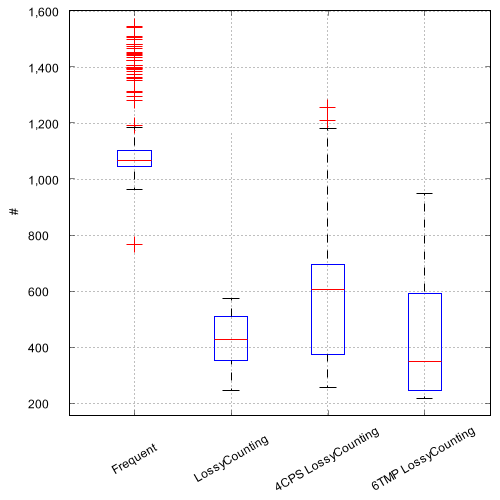}}
  \hfill
  \subfloat[gradual, $exp = 2$, heavy burst]{\includegraphics[width=0.23\linewidth]{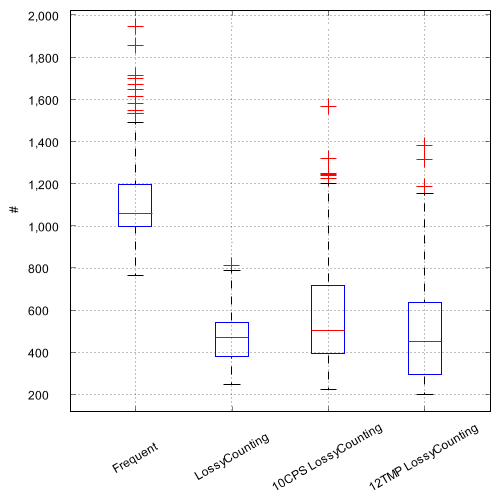}}
  
  \caption{Memory usage (number of counters) for \textit{Frequent}, \textit{Lossy Counting}, \textit{Checkpoint Smoothed} (\textit{CPS}) and \textit{Temporal Smoothed} (\textit{TMP}) algorithms with and without bursts with Zipfian distribution of $exp = 1$ and $exp = 2$}
\end{figure*}

We can observe the standard deviation of the memory usage increase when bursts are introduced.
\textit{Checkpoint Smoothed} uses almost twice as much memory than its base, \textit{Lossy Counting}.
\textit{Temporal Smoothed} uses similar amounts of memory compared to its base but with a greater deviation which becomes more apparent with bursts.

\begin{table}
    \footnotesize
    \centering
    \begin{tabular}{|p{1.75cm}|p{1.05cm}|p{1.05cm}|p{1.05cm}|p{1.05cm}|}
        \hline
            \textbf{Algorithm} & \multicolumn{4}{|c|}{\textbf{Run time}}\\
        \hline
            & gradual $exp=1$ & gradual $exp=1$ burst & gradual $exp=2$ & gradual $exp=2$ burst\\
        \hline
        \hline
            \textit{Frequent}
            & 1355 ms
            & 1168 ms
            & 546 ms
            & 421 ms\\
        \hline
            \textit{Lossy Counting}
            & 1493 ms
            & 1387 ms
            & 616 ms
            & 504 ms\\
        \hline
            \textit{Checkpoint Smoothed}
            & 2089 ms
            & 1854 ms
            & 868 ms
            & 760 ms\\
        \hline
            \textit{Temporal Smoothed}
            & 1979 ms
            & 1784 ms
            & 849 ms
            & 743 ms\\
        \hline
    \end{tabular}
    \caption{Run time for \textit{Frequent}, \textit{Lossy Counting}, \textit{Checkpoint Smoothed} and \textit{Temporal Smoothed} algorithms with and without bursts}
\end{table}

The introduction of bursts does not influence run time significantly.
We can see that our algorithms behave the same way as the state-of-the-art algorithms.
When the exponent becomes higher, the run time becomes lower.

We chose a partition number of $5$ to test our algorithms.
\begin{figure*}
  \centering
  \subfloat[abrupt, $exp = 1$]{\includegraphics[width=0.23\linewidth]{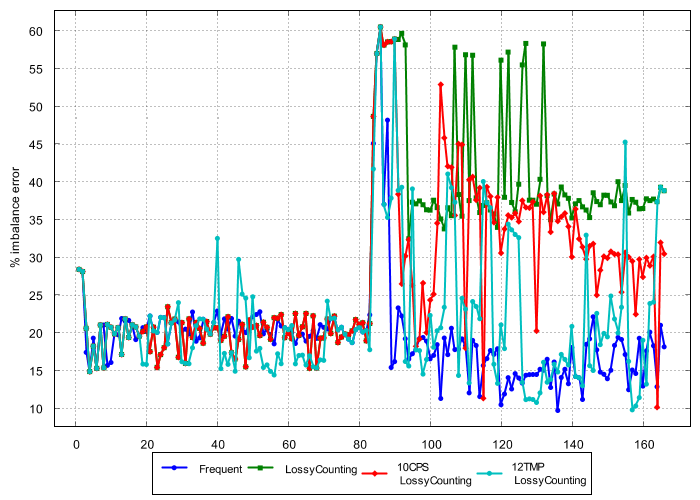}}
  \hfill
  \subfloat[abrupt, $exp = 1$, heavy burst]{\includegraphics[width=0.23\linewidth]{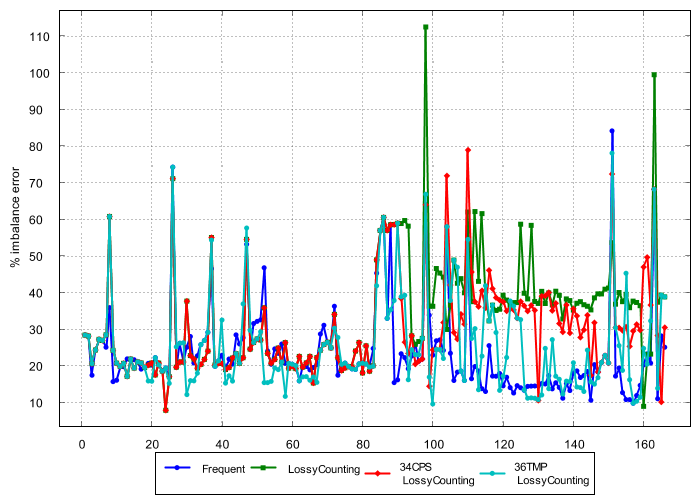}}
  \hfill
  \subfloat[gradual, $exp = 1$]{\includegraphics[width=0.23\linewidth]{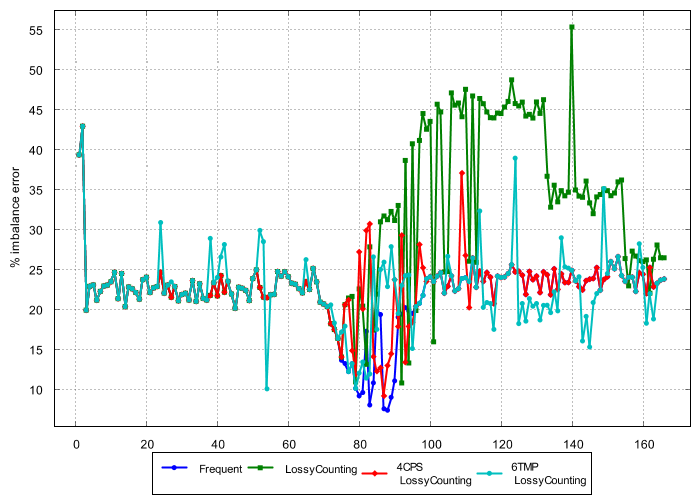}}
  \hfill
  \subfloat[gradual, $exp = 1$, heavy burst]{\includegraphics[width=0.23\linewidth]{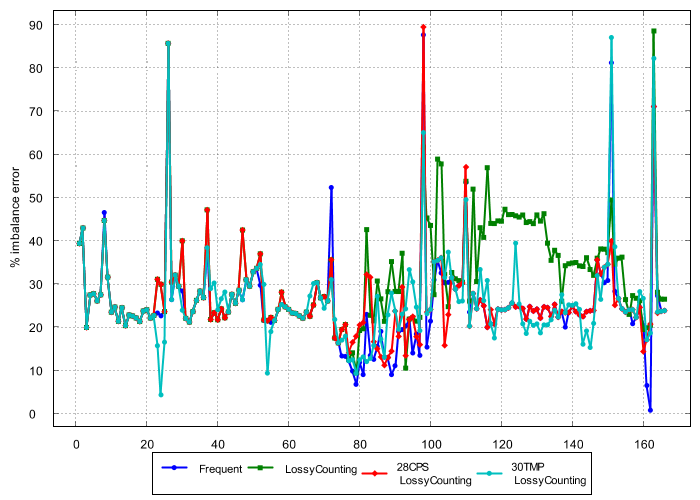}}
  \caption{Percent imbalance error for \textit{Frequent}, \textit{Lossy Counting}, \textit{Checkpoint Smoothed} (\textit{CPS}) and \textit{Temporal Smoothed} (\textit{TMP}) algorithms for abrupt and gradual \textit{concept drifts} with and without bursts with Zipfian distribution of $exp = 1$, using $5$ partitions.}
  \label{fig:tmp_cps_abrupt_exp_1_heavy_burst}
\end{figure*}

An abrupt \textit{concept drift} causes a sudden jump in the imbalance error, which is hard to recover from, but \textit{Temporal Smoothed} can do so quickly.
\textit{Checkpoint Smoothed} is slowly improving.
When heavy bursts are introduced the imbalance caused by the \textit{concept drift} is not significantly greater than before (please note the scale of the diagrams in \autoref{fig:tmp_cps_abrupt_exp_1_heavy_burst}).

The gradual \textit{concept drift} causes dipping in the imbalance error, which can be observed here as well.
The gradual \textit{concept drift} is easier to recover from, and both \textit{Temporal Smoothed} and \textit{Checkpoint Smoothed} can do so quickly.
They can also achieve similar results in the load imbalance in heavy burst scenarios compared to \textit{Frequent}.

In summary, we found that \textit{Temporal Smoothed} and \textit{Checkpoint Smoothed} can react fast to \textit{concept drifts}.
These algorithms improved the reaction time to \textit{concept drifts} compared to the encapsulated \textit{Lossy Counting} algorithm.
The cost of this improvement was the increase in memory usage and run time.

%% file: content/summ.tex
\section{Summary and Future Work}

In this work, we defined \textit{concepts} and \textit{concept drifts} of distributions for \textit{data streams}.
We approached \textit{concept drifts} from different perspectives and showed that the two approaches have the same expressive power, and used this to provide an intuitive way to define \textit{concept drifts} in \textit{data streams}.

We provided a data generator that can handle \textit{data bursts} and \textit{concept drifts}, and showed two methods of measuring the \textit{error of sampling algorithms}.
An optimizer was also included in the \textit{framework} with a random local minimum search that can be used to tune the \textit{hyperparameters} of algorithms.

We introduced two novel algorithms that could react to \textit{concept drifts}, and analyzed them in conjunction with the \textit{state-of-the-art} algorithms using our \textit{framework}.
In our analysis, we found that \textit{Frequent} reacts fastest to \textit{concept drifts}, and our algorithms also show good reaction times.
\textit{Checkpoint Smoothed} could achieve better results than \textit{Frequent} in heavy burst scenarios.

\subsection{Future Work}

We would like to mesure our algorithms on a real world dataset.

We conjecture that \textit{concepts} and \textit{concept drifts} can be extended to any number of underlying probability distributions.

%% file: content/references.tex
\newpage
\printbibliography[heading=bibnumbered]

%% file: content/appendix.tex
\begin{appendices}
\section{Drifts and concepts are equal in expressive power I.}
\label{appendix:section:equal_I}
\begin{proof}
    We construct such $D^{str}$, $\forall c = (c_s, c_e, c_f) \in Con^{str}$ and show that the construction is correct.
    \begin{enumerate}
        \item if $c$ is a changing concept, say $\exists P^c_1, P^c_2$ discrete probability distributions over $K$,
        $\forall i \in \mathcal{D}(c_f), c_f(i) = \Delta(\frac{(i -  c_s)}{c_e - c_s}, P^c_1, P^c_2)$, let 
        \[
            d := (c_e - c_s, \frac{c_e + c_s}{2}, P^c_1, P^c_2) \in D^{str}
        \]
        
        \item if $c$ is a constant concept, say $\exists P^c$ discrete probability distribution over $K$,
        $\forall i \in \mathcal{D}(c_f), c_f(i) = P^c$, let $P^c_x$ be an arbitrary discrete probability distribution over $K$ (it does not matter)
        \[
            d := (0, c_s, P^c_x, P^c) \in D^{str}
        \]
    \end{enumerate}
    For any $i \in [1, 2, \cdots, n], \exists! c = (c_s, c_e, c_f) \in Con^{str}: c_s \leq i \leq c_e$ (based on requirements for $Con^{str}$). We show that for any $c$ concept, our constructed corresponding $d$ drift is correct by showing that $P^{i, str}_G = P^{i, str}_T$
    \begin{enumerate}
        \item if $c$ is a changing concept, with $P^c_1, P^c_2$ discrete probability distributions over $K$,
        $\forall i \in \mathcal{D}(c_f), c_f(i) = \Delta(\frac{(i -  c_s)}{c_e - c_s}, P^c_1, P^c_2)$, from our construction
        \[
            \exists d = (c_e - c_s, \frac{c_e-c_s}{2}, P^c_1, P^c_2) \in D^{str}
        \]
        \[
            s(d) = \frac{c_e + c_s}{2} - \frac{c_e-c_s}{2} \leq i 
        \]
        \[
            i \leq e(d) = \frac{c_e + c_s}{2} + \frac{c_e-c_s}{2}
        \]
        reduces to
        \[
            c_s \leq i \leq c_e
        \]
        Therefore we can use our first generation rule for $d$ drift,
        \begin{equation*}
            \begin{split}
                P^{i, str}_G
                & = \Delta (\frac{(i - s(d))}{c_e - c_s}, P^c_1, P^c_2) \\
                & = \Delta (\frac{(i -  \frac{c_e + c_s}{2} - \frac{c_e-c_s}{2})}{c_e - c_s}, P^c_1, P^c_2) \\
                & = \Delta (\frac{(i -  c_s)}{c_e - c_s}, P^c_1, P^c_2)
                = P^{i, str}_T
            \end{split}
        \end{equation*}
        \item if $c$ is a constant concept, with $P^c$ discrete probability distribution over $K$,
        $\forall i \in \mathcal{D}(c_f), c_f(i) = P^c$
        
        The concepts are non-overlapping, so
        \[
            \forall c' = (c'_s, c'_e, c'_f) \in Con^{str}:
        \]
        \[
            \exists! d' \in D^{str}: c'_s = s(d) \leq e(d) \leq c_e
        \]
        $\exists P^c_x$ (an arbitrary discrete probability distribution over $K$),
        
        if $\exists d = (0, c_s, P^c_x, P^c) \in D^{str}: c_s + \frac{0}{2} = i$, we can use our second generation rule for $d$ drift, that is
        \[
            P^{i, str}_G = P^c = c_f(i) = P^{i, str}_T
        \]
        if $ \exists d = (0, c_s, P^c_x, P^c) \in D^{str}: c_s + \frac{0}{2} < i \land \nexists d' \in D^{str}: c_s + \frac{0}{2} < s(d') \leq i$, we can use our third generation rule for $d$ drift, that is
        \[
            P^{i, str}_G = P^c = c_f(i) = P^{i, str}_T
        \]
    \end{enumerate}
\end{proof}
\section{Drifts and concepts are equal in expressive power II}
\label{appendix:section:equal_II}
\begin{proof}
    We begin by constructing such $Con^{str}$, $\forall d = (d_{len}, d_{mid}, d_{P_1}, d_{P_2}) \in D^{str}$, \\
    
    \begin{enumerate}
        \item if $d_{len} > 0$, let
        \[
            (s(d), e(d), \Delta(\frac{i - s(d)}{d_{len}},d_{P_1},d_{P_2})) \in Con^{str}
        \]
        and if $\nexists d' \in D^{str}: s(d') = e(d) + 1$ and $e(d) \neq n$
        \begin{gather*}
            (e(d)+1, end_d, c_f) \in Con^{str} \\ \textit{, where }  \forall i \in \mathcal{D}(c_f): c_f(i) = d_{P_2} \textit{ and}
        \end{gather*}
        \[
            end_d = 
            \begin{cases}
                n,              & \textit{if } \nexists d' \in D^{str}: d'_{mid} > d_{mid}\\
                s(d') - 1,  & \textit{if } \exists d' \in D^{str}: d'_{mid} > d_{mid}\\
                                & \textit{and } \nexists d'' \in D^{str}: \\
                                & d_{mid} < d''_{mid} < d'_{mid}
            \end{cases}
        \]
         \item if $d_{len} = 0$, let
        \begin{gather*}
            (s(d), end_d, c_f) \in Con^{str} \\
            \textit{, where }  \forall i \in \mathcal{D}(c_f): c_f(i) = d_{P_2} \\
            \textit{, and } end_d \textit{ is the same as before.}    
        \end{gather*}
    \end{enumerate}
    For any $i \in [1, 2, \cdots, n]$, we show that for any $d$ drift, our constructed corresponding $c$ concept is correct by showing that $P^{i, str}_T = P^{i, str}_G$
    \begin{enumerate}
        \item if $\exists d = (len, mid, P_1, P_2): s(d) \leq i \leq e(d)$ and $len > 0$ then, based on the given construction
        \[
            \exists c = (s(d), e(d), \Delta(\frac{i - s(d)}{len}, P_1, P_2))
        \]
        \begin{equation*}
            \begin{split}
            P_T^{i, str} & = \Delta(\frac{i - s(d)}{e(d) - s(d)}, P_1, P_2) \\
            & =  \Delta(\frac{i - s(d)}{len}, P_1, P_2) = P_G^{i, str}
            \end{split}
        \end{equation*}
        \item if $\exists d = (len, mid, P_1, P_2): len > 0 \land e(d) < i \land \nexists d' \in D^{str}: e(d) < s(d') \leq i$ then,
         \begin{enumerate}
            \item if $\exists d' \in D^{str}: d'_{mid} > d_{mid}
            \land \nexists d'' \in D^{str}: d_{mid} < d''_{mid} < d'_{mid}$ then, based on the given construction
            \begin{gather*}
                \exists c = (e(d)+1, s(d') - 1, c_f) \in Con^{str} \\
                \textit{, where } \forall i \in \mathcal{D}(c_f): c_f(i) = P_2
            \end{gather*}
            \item if $\nexists d' \in D^{str}: d'_{mid} > d_{mid}$ then, based on the given construction
            \begin{gather*}
                \exists c = (e(d)+1, n, c_f) \in Con^{str}\\ 
                \textit{, where }  \forall i \in \mathcal{D}(c_f): c_f(i) = P_2
            \end{gather*}
        \end{enumerate}
         In both cases,
        \[
             P_T^{i, str} = c_f(i) = P_2 =  P_G^{i, str}
        \]
        \item if $\exists d = (0, d_{mid}, d_{P_1}, d_{P_2}), e(d) \leq i \land \nexists d' \in D^{str}: e(d) < s(d') \leq i$ then,
        \begin{enumerate}
            \item if $\exists d' \in D^{str}: d'_{mid} > d_{mid}
            \land \nexists d'' \in D^{str}: d_{mid} < d''_{mid} < d'_{mid}$ then, based on the given construction
            \begin{gather*}
                \exists c = (s(d), s(d') - 1, c_f)\\
                \textit{, where } \forall i \in \mathcal{D}(c_f): c_f(i) = d_{P_2}
            \end{gather*}
            \item if $\nexists d' \in D^{str}: d'_{mid} > d_{mid}$ then, based on the given construction
            \begin{gather*}
                \exists c = (s(d), n, c_f) \in Con^{str}\\
                \textit{, where } \forall i \in \mathcal{D}(c_f): c_f(i) = d_{P_2}
            \end{gather*}
        \end{enumerate}
        In both cases,
        \[
            P_T^{i, str} = c_f(i) = d_{P_2} =  P_G^{i, str}
        \]
    \end{enumerate}
\end{proof}
\end{appendices}